\documentclass{article}
\usepackage{arxiv}
\usepackage{graphicx}
\usepackage[utf8]{inputenc} % allow utf-8 input
\usepackage[T1]{fontenc}    % use 8-bit T1 fonts
\usepackage{hyperref}       % hyperlinks
\usepackage{url}            % simple URL typesetting
\usepackage{booktabs}       % professional-quality tables
\usepackage{amsfonts}       % blackboard math symbols
\usepackage{nicefrac}       % compact symbols for 1/2, etc.
\usepackage{microtype}      % microtypography
\usepackage{lipsum}		% Can be removed after putting your text content
\usepackage{natbib}
\usepackage{doi}
\usepackage{citesort}
\usepackage{bm}
\usepackage{amsmath}

\newtheorem{Proposition}{Proposition}
\newtheorem{Corollary}{Corollary}
\newtheorem{Definition}{Definition}
\newtheorem{Example}{Example}
\newtheorem{proof}{Proof}

\title{Petri Net Modeling for Ising Model Formulation in Quantum Annealing}

\author{ \href{https://orcid.org/0000-0003-1583-2523}{\includegraphics[scale=0.06]{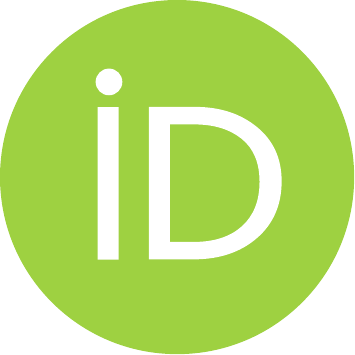}\hspace{1mm}Morikazu Nakamura}\\
	Computer Science and Intelligent Systems\\
	Faculty of Engineering\\
	University of the Ryukyus\\
	Okinawa 903-0213, Japan\\
	\texttt{morikazu@ie.u-ryukyu.ac.jp} \\
	\And
	Kohei Kaneshima \\
	Information Engineering Course\\
	Graduate School of Engineering and Science\\
	University of the Ryukyus\\
	Okinawa 903-0213, Japan\\
	\texttt{k208580@ie.u-ryukyu.ac.jp} \\
	\And
	Takeo Yoshida \\
	Computer Science and Intelligent Systems\\
	Faculty of Engineering\\
	University of the Ryukyus\\
	Okinawa 903-0213, Japan\\
	\texttt{tyoshida@ie.u-ryukyu.ac.jp} \\
}

\begin{document}
% Abstract (Do not insert blank lines, i.e. \\) 
\maketitle

\begin{abstract}
Quantum annealing is an emerging new platform for combinatorial optimization, requiring an Ising model formulation for optimization problems.
The formulation can be an essential obstacle to the permeation of this innovation into broad areas of everyday life.
Our research is aimed at the proposal of a Petri net modeling approach for an Ising model formulation.
Although the proposed method requires users to model their optimization problems with Petri nets, this process can be carried out in a relatively straightforward manner if we know the target problem and the simple Petri net modeling rules.
With our method, the constraints and objective functions in the target optimization problems are represented as fundamental characteristics of Petri net models, extracted systematically from Petri net models, and then converted into binary quadratic nets, equivalent to Ising models.
The proposed method can drastically reduce the difficulty of the Ising model formulation.
\end{abstract}

% Keywords

\section{Introduction}
Quantum annealing is a metaheuristic for combinatorial optimization in quantum mechanics research\citep{PhysRevE.58.5355},\citep{qa_science}.
This new approach solves unconstrained optimization problems formulated as a Hamiltonian by evolving the time-dependent Schr\"{o}dinger equation from a quantum mechanical superposition of all possible states to its ground state in physical systems\citep{AIP_Nishimori}.
A quantum annealing machine is a special-purpose quantum computer that solves combinatorial optimization problems.
D-wave is the first commercial quantum annealing machine\citep{dwave}.

Combinatorial optimization contributes to a more effective and reasonable life by minimizing the total costs or maximizing the benefits under certain constraints.
Although sophisticated solvers can efficiently solve small problems, many practical problems are computationally intractable. 
Such problems are formally characterized as NP-hard; that is, there are no polynomial-time algorithms to solve the problems. 
Many researchers in computer science and operations research believe that there are no polynomial-time algorithms for NP-hard problems and continue to develop efficient heuristic algorithms to solve larger problems\citep{10.5555/574848}.

Quantum annealing is a metaheuristic that does not restrict the target problems.  
Thus far, some combinatorial optimization problems, but not many, have been formulated as Ising or Quadratic Unconstrained Binary Optimization (QUBO) models and solved using annealing machines\citep{10.3389/fphy.2014.00005}, \citep{osti_1756458}, \citep{8643733}, \citep{NSP}.
Theoretical and experimental studies have contributed to improving the performance of annealing processes \citep{PhysRevApplied.14.014100}, \citep{PhysRevLett.123.120501}, \citep{PhysRevLett.126.070505}, \citep{Hauke_2020}.
Useful software tools have been developed to utilize quantum annealing machines \citep{pyqubo}.
Digital implementations of the quantum annealing process have also emerged because this new optimization technique has the potential to outperform traditional meta-heuristic algorithms.
A graphics processing unit (GPU) based quantum annealing simulator has shown high-performance optimization by utilizing spatial and temporal parallelism during the annealing process \citep{9057502}.
Therefore, we have high hopes for this new innovative optimizer. 

However, obstacles to quantum annealing usability include the hardness of the parameter tuning and formulation of the target problems as Ising or QUBO models.
The constraints of the original problem and objective functions are represented as a penalty function, minimized as an unconstrained optimization problem.
Search processes can reach infeasible solutions, and
to avoid such situations, it is necessary to carefully tune the parameters in pre-performing annealing process.
The difficulty of the Ising model formulation is also an essential obstacle to the usability of quantum annealing.
The new annealer requires users to create quadratic penalty functions, Ising, or QUBO models, for target optimization problems.

This study focuses solely on the latter problem, which is the difficulty of formulating the Ising and QUBO models.
To overcome this problem, we propose a method for generating Ising or QUBO models from Petri net models representing the target optimization problems.
Our proposal is a Petri net modeling approach in which, once we model the optimization problem with Petri nets, we can systematically obtain the Ising or QUBO formulation.

A Petri net is a mathematical and graphical modeling language for various systems, such as computer networks, manufacturing systems, transportation networks, biological systems, agricultural production processes, and other applications\citep{24143}.
We only need the domain knowledge of the target systems for modeling because the rules and components of Petri nets are simple and intuitive.
Petri net modeling is quite efficient for representing integer linear programming problems because Petri nets have mathematical forms, denoted as linear expressions\citep{refId0}.
The authors proposed an algorithm for generating mixed-integer linear programming (MILP) based on Petri net models for practical scheduling and optimal resource assignment problems\citep{MorikazuNAKAMURA2019}. 
Our algorithm realizes the automatic formulation of MILPs for given combinatorial optimization problems. 

This study extends our MILP version to quantum annealing models, where quadratic functions of binary variables need to be formulated.
As far as we know, this is the first study to apply Petri net modeling to the quantum annealing.
In this paper, we introduce a new notation called \textit{binary quadratic nets}, which denote Ising or QUBO models, respectively.
Our method incrementally targets binary quadratic nets from problem-domain Petri net models that represent target optimization problems.

This paper is organized as follows: Section 2 summarizes the basic knowledge on quantum annealing and Petri nets. 
Section 3 introduces a new class of Petri nets, binary quadratic nets, for representing Ising or QUBO models.
Section 4 proposes a method for formulating Ising or QUBO models from problem-domain Petri nets and presents some examples.
Finally, in Section 5, we present some concluding remarks and areas of future tasks.
%%%%%%%%%%%%%%%%%%%%%%%%%%%%%%%%%%%%%%%%%%

\section{Preliminaries}
This section provides the basic definitions and notations of Ising models and Petri nets.

\subsection{Ising and QUBO Models}
Quantum annealing is a metaheuristic for solving combinatorial optimization problems, where we try to find a value $+1$ or $-1$ for each of the Ising variables $\boldsymbol s = (s_1, s_2, ..., s_N)$ such that the following Hamiltonian of the Ising models gives the lowest energy state:
\begin{equation}
H(\boldsymbol s) =  \sum_{i=1}^{N}h_is_i + \sum_{i<j} J_{i,j}s_is_j, \label{eqn:ising_model}
\end{equation}
where $h_i$ is the magnetic field coefficient at site $s_i$, and $J_{i, j}$ is the interaction coefficient between $s_i$ and $s_j$.
Ising variables correspond to discrete variables in metaheuristics.
The lowest energy state, called the {\it ground state} of $H$, provides the optimal solution.
Figure \ref{fig:Ising_bqn_model} depicts a two-dimensional array Ising model.
A circle represents a spin and is connected to four neighbor spins.
The arrows $\uparrow$ and $\downarrow$ represent the \textit{up} spin (+1) and \textit{down} spin (-1), respectively.
\begin{figure}[h]
 \begin{center}
 \includegraphics[scale=0.5]{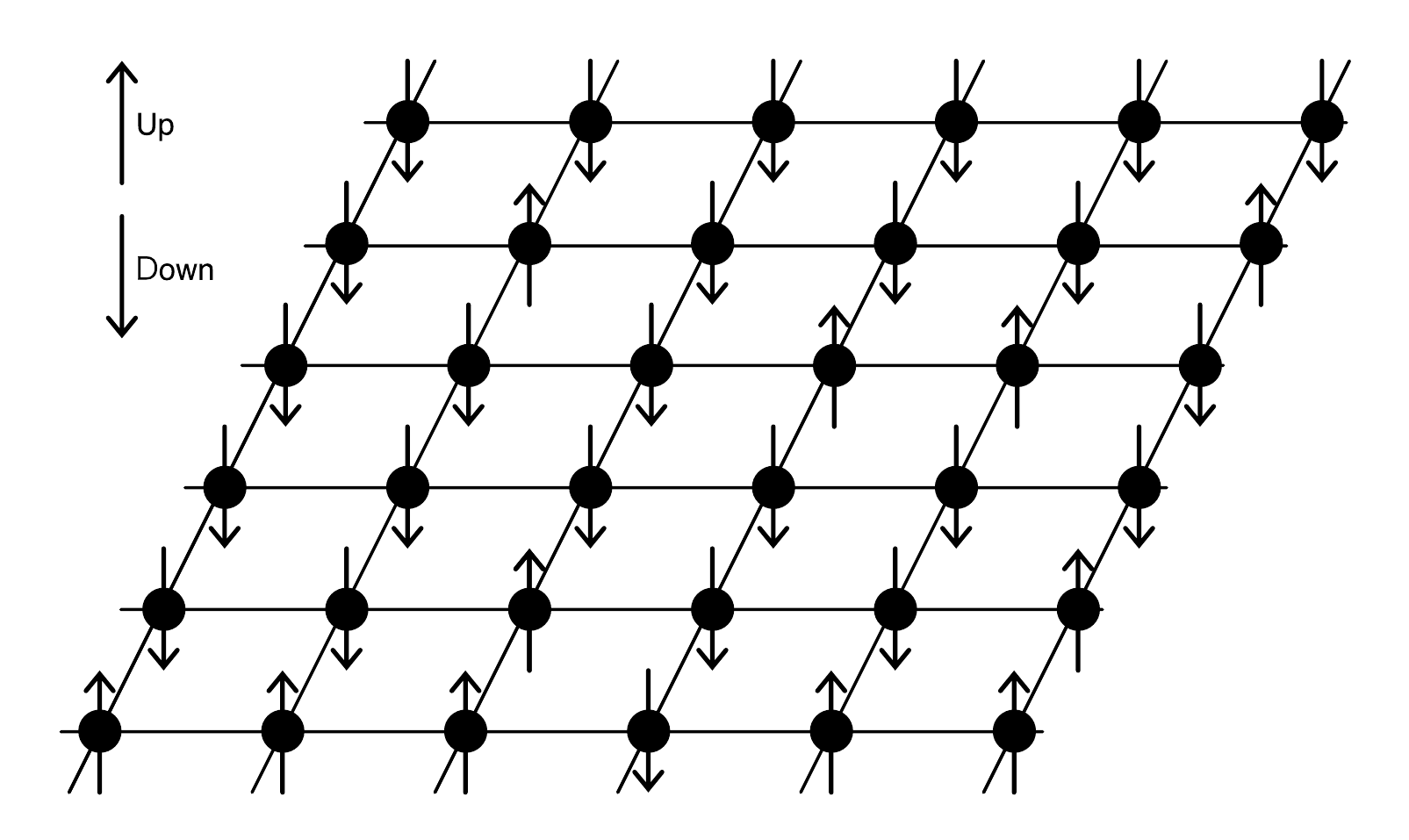}
 \caption{2D Array Ising Model}
 \label{fig:2d_ising}
 \end{center}
\end{figure}

A Hamiltonian may include an objective function, where the lower energy on the terms of the function contributes to a smaller objective value in minimization problems.
For maximization problems, we reverse the sign of the Ising model parameters for the objective function and then minimize it.
Constraints in optimization problems should also be composed of sets of terms in the Hamiltonian, where the lowest energy on the terms leads to the satisfaction of the corresponding constraints.
Therefore, the annealing process needs to obtain feasible solutions that satisfy all constraints with a good quality of the objective value by exploring the lowest energy state of the total Hamiltonian. 

Note that we can easily replace the type of decision variables from $\{+1, -1\}$ to $\{0, 1\}$, and vice versa.

\subsection{Petri Net Fundamentals}

A Petri net $\mathit{PN} = (P, T, F, \mathbf{M}_0)$ is a directed bipartite graph $N = (P, T, F)$ with initial marking $\mathbf{M}_0$.
The bipartite graph has two types of vertex sets: a set of places $P=\{P_1, P_2, ..., P_{|P|}\}$ and a set of transitions $T=\{T_1, T_2, ..., T_{|T|}\}$. 
The arc function $F:(P\times T) \cup (T\times P) \rightarrow \mathbb{N}$ defines the connection from places to transitions, and vice versa.
The returned values of $F$ indicate the number of removed tokens from the starting place when $(P\times T) \rightarrow \mathbb{N}$ or the number of generated tokens to the ending place when $(T\times P) \rightarrow \mathbb{N}$ upon the firing of the transition.

Instead of an arc function, a matrix representation is often used for the connection.
Here, $\mathit{Pre}$ and $\mathit{Post}$ show the incident matrix of size $|P|\times |T|$ to indicate the connection from places to transitions and from transitions to places, respectively.
It should be noted that $\mathit{Pre}(P_i, T_j) = F((P_i, T_j))$ and $\mathit{Post}(P_i, T_j) = F((T_j, P_i))$.

Tokens are located in places such that the token distribution on $P$ represents the status of the modeled system.
The vector $\mathbf{M}_k$ of size $|P|$ represents the number of tokens in each place $P_i \in P$ at step $k$.
In addition, $\mathbf{M}_0$, called the \textit{initial marking}, is the marking at step $0$, which corresponds to the initial status of the system.
Therefore, the Petri net $PN = (N, \mathbf{M}_0)$ represents the structure and initial states of the modeled system.

The transition $T_j \in T$ is {\it enabled} at step $k$ only when $M_k(P_i)$ $\ge \mathit{Pre}(P_i, T_j)$ for each $P_i$ in $^\bullet T_j$, 
where $^\bullet T_j$ and $T_j^\bullet$ denote the set of input and output places of $T_j$, respectively.
Similarly, $^\bullet P_i$ and $P_i^\bullet$ are the sets of input and output transitions of $P_i$, respectively.
An enabled transition $T_j$ can fire.
Firing implies the occurrence of an event in the system.
The firing of transition $T_j$ removes $\mathit{Pre}(P_i, T_j)$ tokens from each $P_i \in \ ^\bullet T_j$ and adds $\mathit{Post}(P_i, T_j)$ tokens to each $P_i \in T_j^\bullet$.
Firing changes the token distribution, which indicates the change in status of the system by the corresponding event.
The following mathematical form shows the change in status caused by the firing of transitions at step $k$:
\begin{eqnarray}
 \mathbf{M}_{k+1} & = & \mathbf{M}_{k} - \mathit{Pre} \cdot \mathbf{X}_k + \mathit{Post} \cdot \mathbf{X}_k \nonumber\\
   & = & \mathbf{M}_{k} + (\mathit{Post}-\mathit{Pre}) \cdot \mathbf{X}_k, \label{eqn:state},
\end{eqnarray}
where $\mathbf{X}_k$ is a vector of size $|T|$ showing the number of times each $T_j$ fires at step $k$, and $\mathbf{X}_k$ is called the {\it firing count vector} at step $k$.

For a quantitative analysis of the dynamical behavior of a system, {\it time} was introduced to Petri nets \citep{aalst}.
The timing methods can be categorized into three types: firing duration (FD), holding duration (HD), and enabling duration (ED) \citep{aalst}.
The FD assigns time to transitions to represent the firing duration. 
The HD is referred to as place time Petri nets, where tokens are unavailable for firing for a particular period after being located in the place.
In the last method, i.e., the ED, a transition cannot be fired for a given period after being enabled. 
Although other types can be allowed with our method, this paper focuses on timed Petri nets with the FD.

A timed Petri net is a six-tuple ${\mathit TPN} = (P, T, F, TS, FD, \mathbf{M}_0)$, where $TS$ is the set of time values, and $FD: T \rightarrow TS$ is a function that returns the firing duration time of transition $T_j\in T$.
In this study, we assume that $TS =\mathbb{N}$ and is a set of natural numbers.
A timestamp is also attached to tokens to record the token generation time. 
In the timed Petri net, the transition $T_j$ is enabled at time $k$ when each input place $P_i$ of $T_j$ has more than or equal to $F(P_i, T_j)$ tokens, and its time stamp is no more than $k$.
By firing $T_j$ at time $k$, the marking is changed according to the same rule as the Petri net described above, except that we attach the time stamp $k + FD[T_j]$ to each output token.

A colored Petri net is an extension of ordinary Petri nets, where we can introduce values (colors) to tokens, guard functions to transitions for their firing conditions, and arc functions to output arcs of transitions to calculate the colors of the output tokens\citep{cpn}.
Therefore, tokens become much more informative and allow us flexible and efficient modeling. 

This extended Petri net is denoted by $\textit{CPN}=(\Sigma, P, T, F, V, C, G, E, \mathbf{M}_0)$, where $P$, $T$, and $F$ represent a set of places, transitions, and arcs, respectively.
$\Sigma$ shows a set of colors, and $C: P\rightarrow \Sigma$ is a color function for the places. 
Only tokens with colors specified by $C(P_i)$ are located in place $P_i$.
In addition, $V$ is a set of arc variables, where $\textrm{Type}(v) \subseteq \Sigma$ for $v\in V$.
Moreover, $\textrm{Type}(v)$ is the type of variable $v$.
Here, $E$ denotes an arc function where $E((T_j, P_i))$ returns tokens on $C(P_i)$ for each output place $P_i$ based on the binding values to each of the arc variables.
The term $G$ represents a guard function that returns a Boolean value to determine whether the corresponding transition is enabled.
A marking $\mathbf{M}_k$ at step $k$ is a mapping of multiple sets of $C(P_i)$ to each place.
Transition $T_j$ is enabled at marking $\mathbf{M}_k$ with binding $b$ when for each input place $P_i$ of $T_j$, 
\begin{equation}
M_k(P_i) \ge E((P_i, T_j))(b)
\end{equation}
where $E((P_i, T_j))(b)$ denotes the result of the arc function for arc $(P_i, T_j)$ when we bind $b(v)$ for each variable $v$ in $E((P_i, T_j)).$
At marking $\mathbf{M}_k$, $\mathbf{M}_{k+1}$ is generated by firing $T_j$ with binding $b$:
\begin{equation}
M_{k+1}(P_i) = M_k(P_i) -E((P_i, T_j))(b) + E((T_j, P_i))(b).
\end{equation}
Colored Petri nets are extremely powerful in the sense that complicated systems can be easily modeled.

%%%%%%%%%%%%%%%%%%%%%%%%%%%%%%%%%%%%%%%%%%
\section{Binary Quadratic Nets}

This section introduces a new Petri net class called \textit{binary quadratic nets} to represent the Ising and QUBO models.
The Ising model is a mathematical model in statistical mechanics.
The model is a graph in which the vertices correspond to spins, and the edges interact between spins.
Petri nets are also suitable for expressing Ising models because their fundamental components, places, and transitions can naturally represent the states of the spins and their interactions.
A token in a place can show the corresponding spin state when we introduce colors $\{-1, +1\}$ to tokens.
For QUBO models, the color type of tokens becomes $\{0, 1\}$.

In our Petri net modeling-based Ising and QUBO model formulation, Petri net models representing target optimization problems were converted into the corresponding binary quadratic nets. 
In other words, our proposed method is a net transformation from problem-domain Petri nets into binary quadratic nets.
Binary quadratic nets are entirely equivalent to Ising or QUBO models.

\subsection{Formal Definition}

\begin{Definition} \textbf{(Binary Quadratic Net)}
A binary quadratic net is a colored Petri net denoted by $\textrm{BQN} = (\hat{\Sigma}$, $\hat{P}$, $\hat{T}$, $\hat{F}$, $\hat{C}$, $\hat{w})$:
\begin{align}
 \hat{\Sigma} &\in \{\{-1, +1\}, \{0, 1\}\} \\
\hat{P} & = \{\hat{p}_1, \hat{p}_2, ..., \hat{p}_n\} \\
\hat{T} &= \{\hat{t}_{i, j} \mid \hat{p}_i, \hat{p}_j \in \hat{P}, \{\hat{p}_i, \hat{p}_j\} = \hat{t}_{i, j}^\bullet = ^\bullet\!\! \hat{t}_{i, j}, i<j\}\\
\hat{F} &: (\hat{P} \times \hat{T})\cup (\hat{T}\times \hat{P}) \rightarrow 
\begin{cases}
 1 & \textrm{if } (\hat{p}_i, \hat{t}_{i,j}) \textrm{ and } (\hat{t}_{i,j}, \hat{p}_j), \hat{t}_{i,j} \in \hat{T} \\
 0 & \textrm{otherwise}
\end{cases} \\
\hat{C}&(\hat{p}) \in \hat{\Sigma}, \forall \hat{p} \in \hat{P}\\
\hat{w} &: \hat{P}\cup \hat{T} \rightarrow \mathbb{R}
\end{align}
\label{def:bqnet}
\end{Definition}

There is one token in each place, where a place corresponds to a spin variable in the Ising models and a binary variable in the QUBO models, respectively.
For simplicity, we denote the value of a variable associated with place $\hat{p}_i$ by $M(\hat{p}_i)$, that is, $M(\hat{p}_i) \in \{-1, +1\}$ in the Ising models or $M(\hat{p}_i) \in \{0, 1\}$ in the QUBO models. 
Note that in the QUBO model, $M(\hat{p}_i)=0$ indicates that the place $\hat{p}_i$ includes a token with color $0$.

Current quantum annealing platforms have various physical topologies, chimera graphs for D-waves, and three-dimensional lattices for CMOS annealers \citep{DaisukeOKU20192018EDP7411}. 
Thus, we need to embed the logical Ising model to a specific graph topology to run the annealing process.
However, we do not treat the embedding process in this study because we can use existing embedding algorithms for each platform.
Moreover, Fujitsu Digital Annealer allows fully connected topologies.
For reference, Fig. \ref{fig:Ising_bqn_model} depicts the binary quadratic net model for the 2D array Ising model shown in Fig.\ref{fig:2d_ising}.
\begin{figure}[h]
 \begin{center}
 \includegraphics[scale=0.65]{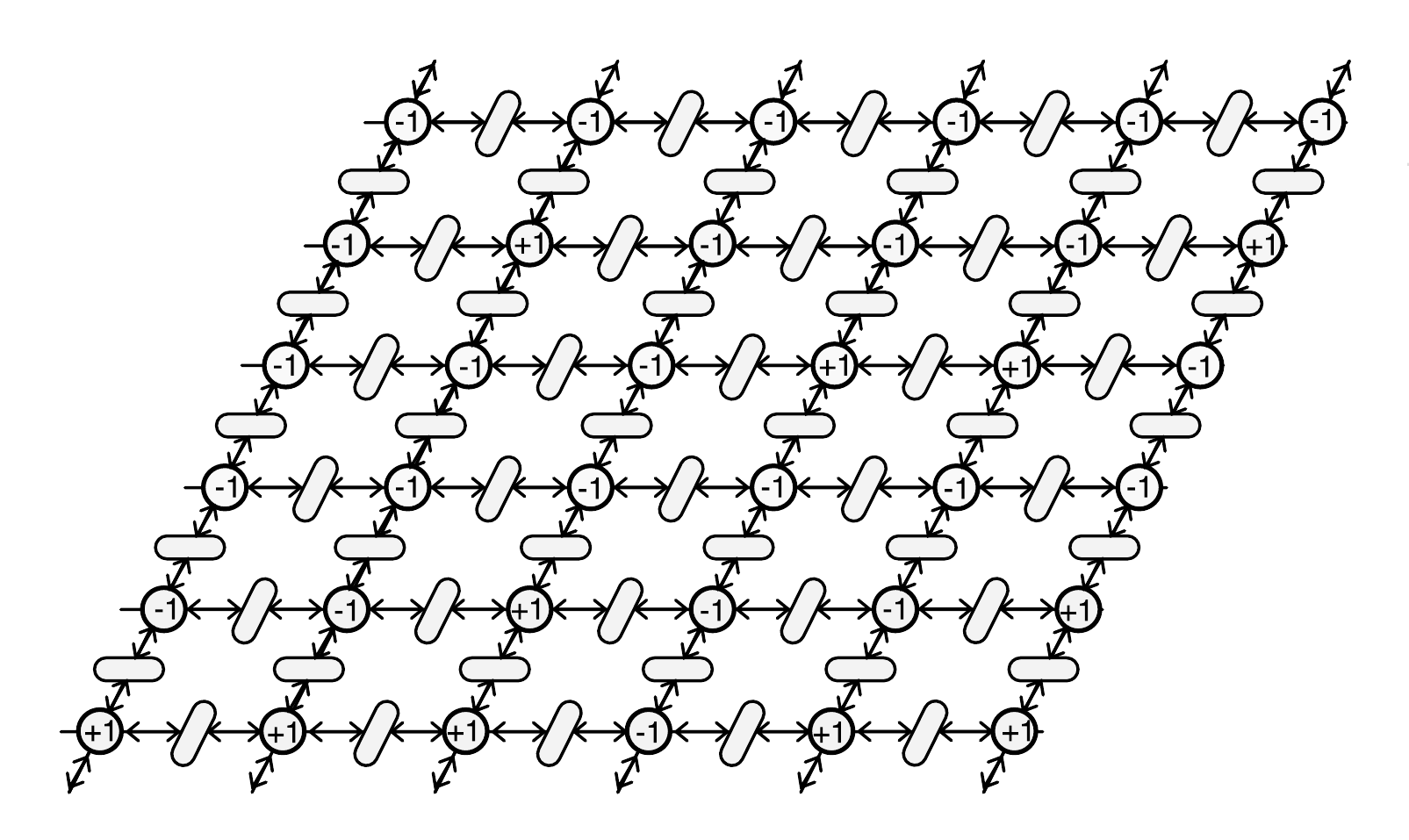}
 \caption{Binary Quadratic Net for 2D-Array Ising Model shown in Fig. \ref{fig:2d_ising}}
 \label{fig:Ising_bqn_model}
 \end{center}
\end{figure}

To combine binary quadratic nets and Ising models, we introduce a measure to represent the state of binary quadratic nets with a marking $\mathbf{M}$.
We define this as an energy function of binary quadratic nets corresponding to the Hamiltonian in the Ising models.

\begin{Definition}
\textbf{(Energy Function of Binary Quadratic Net)} 
For a binary quadratic net $\textrm{BQN} = (\hat{\Sigma}, \hat{P}, \hat{T}, \hat{F}, \hat{C}, \hat{w})$ with a marking $\mathbf{M}$, the energy function is defined as follows:
\begin{equation}
H_{BQN}(\mathbf{M}) = \sum_{\hat{p}_i \in \hat{P}} w(\hat{p}_i) M(\hat{p}_i)  + \sum_{\hat{t}_{i, j} \in \hat{T}} w(\hat{t}_{i, j})M(\hat{p}_i)M(\hat{p}_j).
\end{equation}
\label{def:enegyFunction}
\end{Definition}
The first summed terms show the energy derived from token existence at each place.
The second summed term represents the energy from the interaction between tokens on both sides of each transition. 
The weight parameters $w(\hat{p}_i)$ and $w(\hat{t}_{i, j})$ are defined for $\hat{p}_i$ and $\hat{t}_{i, j}$, respectively.
The energy function in Definition \ref{def:enegyFunction} is equivalent to the Ising model shown in (\ref{eqn:ising_model}), and the energy function is uniquely defined from the corresponding binary quadratic net.
Although it is straightforward, we simply summarize this fact as a proposition for the readability of the remaining parts.
\begin{Proposition}
For an Ising model, we have an equivalent binary quadratic net $BQN$ in the sense that the energy function $H_{BQN}(\mathbf{M})$ is equivalent to the Hamiltonian of the Ising model.
\label{prop:equivalence}
\end{Proposition}
In our binary quadratic nets, transitions represent interactive relations between tokens on both sides.
There may be numerous interaction types.
We choose interaction types carefully depending on the target optimization problems.
In the Appendix, we summarize the primitive interaction for QUBO and the Ising model, respectively.
Table \ref{atbl:quboip} shows the interaction primitives $I^{qubo}_i, i=0, 1,..., 15$, for two binary variables in $\{0, 1\}$, 
where $(0, 0), (0, 1), (1, 0)$, and $(1, 1)$ express all possible combinations of the two binary variables. In addition,
$I^{qubo}_1, I^{qubo}_7, I^{qubo}_8, I^{qubo}_9$ correspond to the well-known logical functions AND, XOR, OR, and NOR, respectively.
Moreover, $I^{qubo}_{0}, I^{qubo}_{15}$ indicate inconsistency and a tautology, respectively.
Table \ref{atbl:isingip} lists the interactions for the Ising model converted from those in Table \ref{atbl:quboip} by applying the following relation.
\begin{align}
M(\hat{p}_i)^{ising} &= 2 M(\hat{p}_i)^{qubo} -1, \label{eqn:ising_qubo} \\ 
M(\hat{p}_i)^{qubo} &= (M(\hat{p}_i)^{ising} + 1)/2, \label{eqn:qubo_ising} 
\end{align}
where $M(\hat{p}_i)^{ising}$ and $M(\hat{p}_i)^{qubo}$ denote the marking in the Ising and QUBO models, respectively.

\subsection{Binary Quadratic Net Examples}

As we describe in Proposition \ref{prop:equivalence}, binary quadratic net models are equivalent to the Ising models.
We show the binary quadratic net models for well-known graph partitioning and minimum vertex cover problems.
These problems are straightforwardly modeled with binary quadratic nets and formulated as marking problems in such nets.

\begin{Example}
\textbf{Minimum Vertex Cover:} 
For a given undirected graph $G=(V, E)$, with vertex set $V$ and edge set $E$, a vertex cover set satisfies the condition such that every edge in $E$ is incident on a vertex in the cover set.
The problem is to find the minimum vertex cover set.
The problem is known as NP-hard\citep{10.5555/574848}.

We choose the QUBO model for the problem, where we generate the place set and transition set of the binary quadratic net under the one-to-one correspondence with the vertex set and the edge set, respectively.
We express a vertex cover by a marking, where we place a token with a color of 1 to show that the corresponding vertex is in the vertex cover, and where a color of 0 is not in the cover.
To satisfy the feasibility of the vertex cover, we need to avoid a token with a color of 0 in both input places of each transition because the situation does not cover the corresponding edge.
We formulate the feasibility of the condition corresponding to $I_8^{qubo}$ in Table \ref{atbl:quboip} as follows:

\begin{align}
H_\textrm{constraint}(\mathbf{M}) &= \sum_{i<j, \hat{p}_i^\bullet \cap \hat{p}_j^\bullet \ne \emptyset} (1-M(\hat{p}_i))(1-M(\hat{p}_j)), \label{eqn:vertexcover_constraint}
\end{align}

To minimize the objective function, we attempt to reduce the number of color 1 tokens in $\mathbf{M}$.
Therefore, the following penalty function is suitable because only color 1 tokens increase the penalty.
\begin{align}
H_\textrm{cost}(\mathbf{M}) &= (\sum_{i=1}^{|P|} M(\hat{p}_i)) \label{eqn:vertexcover_cost} 
\end{align}

Based on the superposition principle, we have the total formulation based on our binary quadratic net models by combining the binary quadratic nets corresponding to (\ref{eqn:vertexcover_constraint}) and (\ref{eqn:vertexcover_cost}):
\begin{align}
H_\textrm{total}(\mathbf{M}) &= A \cdot H_\textrm{constraint}(\mathbf{M}) + B \cdot H_\textrm{cost}(\mathbf{M})
\end{align}
where $H_\textrm{constraint}(\mathbf{M})$ counts the number of transitions, both of which have a 0 color token.
In addition, $H_\textrm{cost}(\mathbf{M})$ denotes the number of places colored with 1 to minimize the vertex cover.
Moreover, $A$ and $B$ show the parameters for a trade-off between the constraint and the objective function.
The formulation is equivalent to the QUBO model in \citep{10.3389/fphy.2014.00005}.

Consider the graph shown in Fig. \ref{fig:undirected_graph} as a problem instance of the vertex cover problem. 
We can formulate the problem instance as a marking problem in the corresponding binary quadratic net in Fig. \ref{fig:bqnet_vertex_cover}.
Marking $\mathbf{M}=(1, 0, 1, 0, 0, 1, 0, 1)$ is a feasible solution to the instance.
 
\begin{figure}[h]
 \begin{center}
 \includegraphics[scale=0.3]{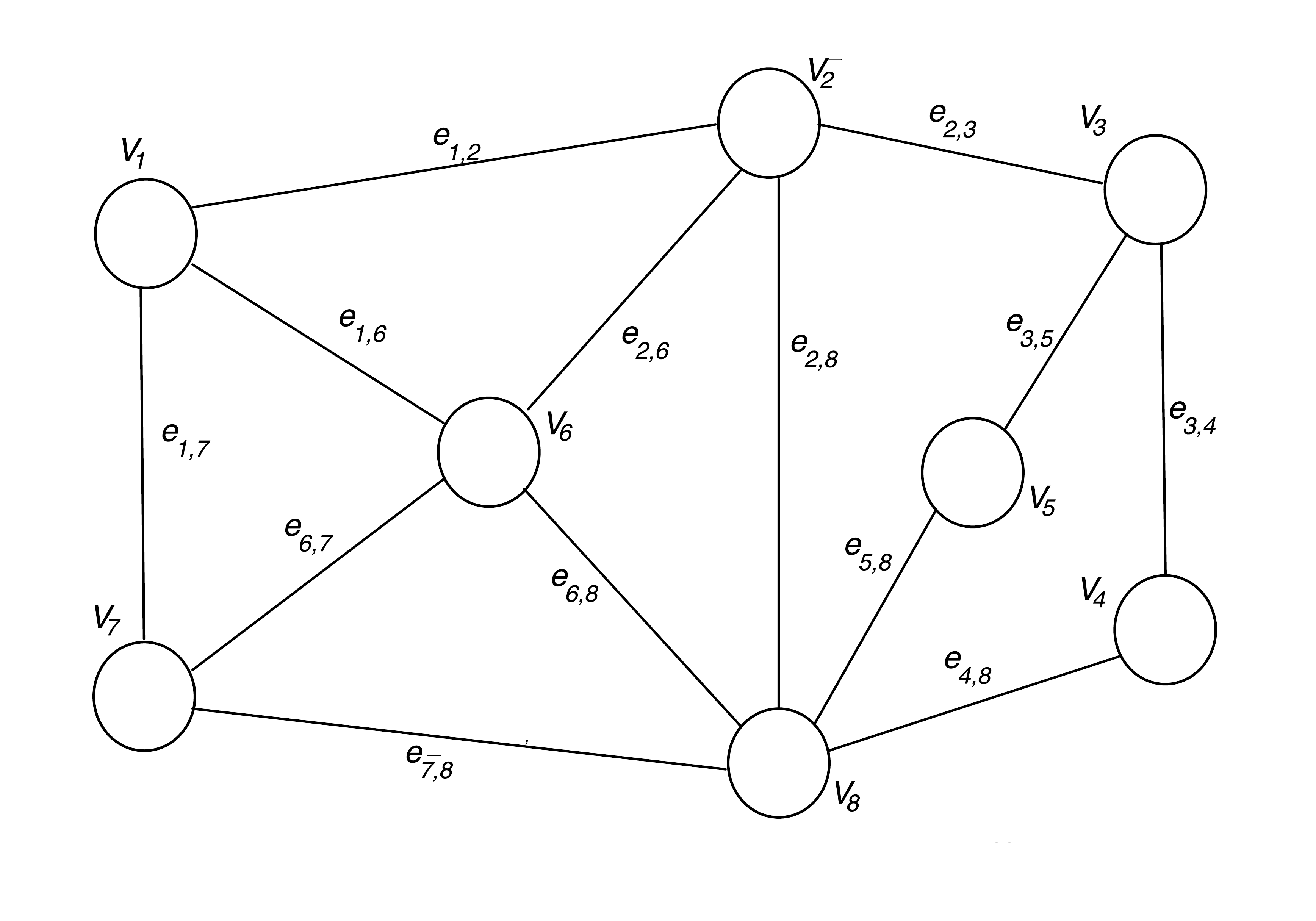}
 \caption{Example of Undirected Graphs}
 \label{fig:undirected_graph}
 \end{center}
\end{figure}

\begin{figure}[h]
 \begin{center}
 \includegraphics[scale=0.3]{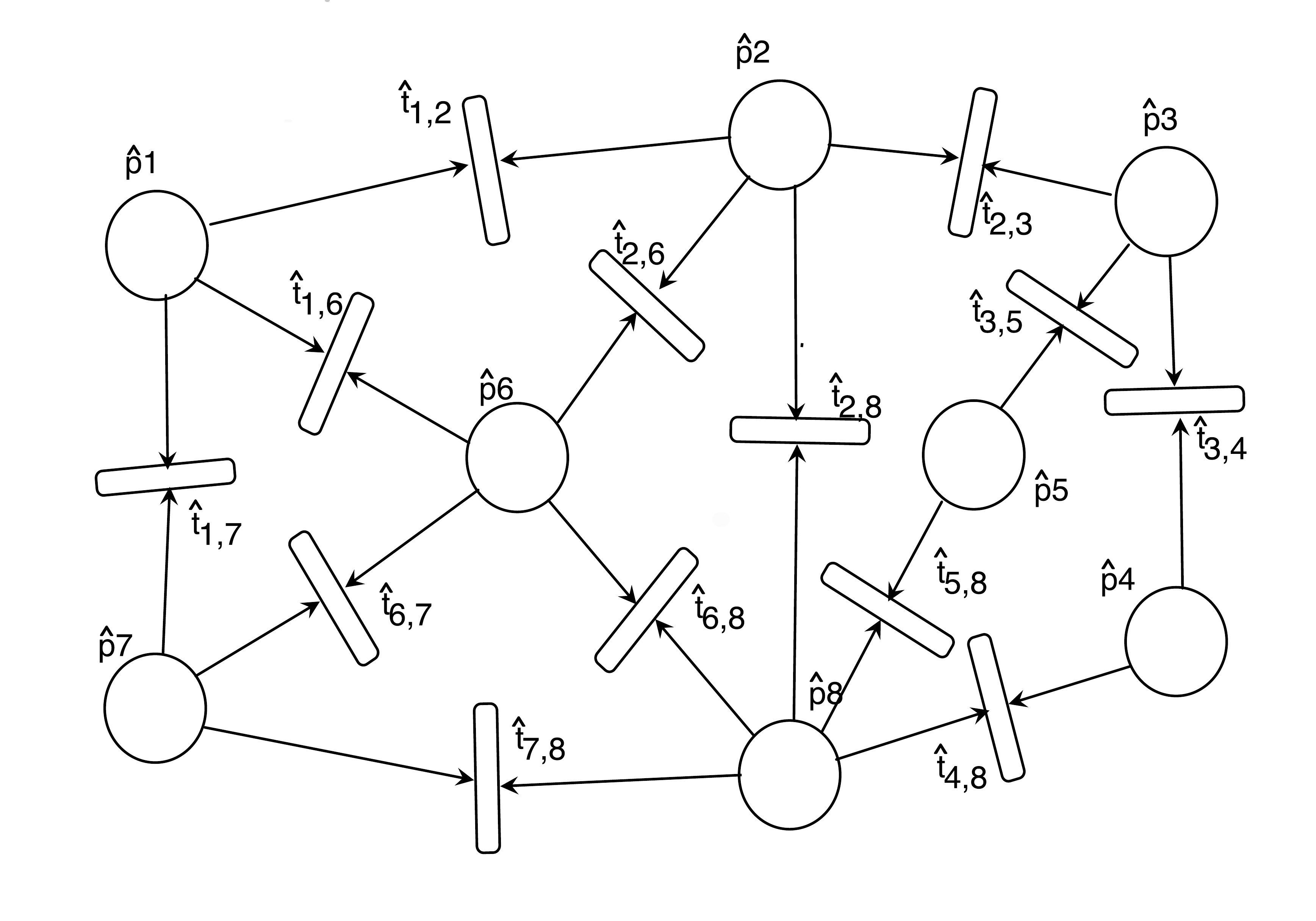}
 \caption{Binary Quadratic Net for Vertex Cover Problem Instance shown in Fig.\ref{fig:undirected_graph}}
 \label{fig:bqnet_vertex_cover}
 \end{center}
\end{figure}

\end{Example}

\begin{Example}
\textbf{Graph Partitioning:} 
Graph partitioning is an optimization problem, which is explained as follows:
Let us consider an undirected graph $G =(V, E)$ with vertex set $V$ and edge set $E$, where $|V|$ is the number of vertices. 
The problem is partitioning $V$ into two subsets whose sizes are equal to $|V|/2$, minimizing the number of edges between the subsets.
The problem is known as NP-hard\citep{10.5555/574848}.

We reduce the problem to a marking problem in binary quadratic nets, where each vertex in the given graph corresponds to a place in the binary quadratic net, and the marking such that $M(\hat{p}_i) \in \{-1, +1\}, \forall \hat{p}_i \in \hat{P}$ shows the partition.
The objective function minimizes the number of edges connecting the two partitioned groups.
The transitions have a one-to-one correspondence with the edges.
To minimize the objective function, we want to reduce the different color tokens in pairs of places that share the output transition.
We design the objective function with $I^{Ising}_6$ in Table \ref{atbl:isingip} as follows:
\begin{align}
H_\textrm{cost}(\mathbf{M}) &= \sum_{i<j, \hat{p}_i^\bullet \cap \hat{p}_j^\bullet \ne \emptyset} \frac{1-M_0(\hat{p}_i)M_0(\hat{p}_j)}{2} \label{eqn:graph_partitioning_cost} 
\end{align}
Here, $I^{Ising}_6$ outputs 1 only when both places have different color tokens, that is, $XOR$.

Second, we consider the constraint of the graph-partitioning problem and the equality in the vertex size of both partitioned groups.
Concerning $\{-1, +1\}$ logic, we can design the following energy function for the constraint because minimizing the function leads to a satisfaction of the constraint:
\begin{align}
H_\textrm{constraint}(\mathbf{M}) &= (\sum_{i=1}^{|P|} M(\hat{p}_i))^2 \label{eqn:graph_partitioning_constraint} 
\end{align}
The following penalty function shows the total energy function for graph partitioning and is equivalent to the Ising model presented in \citep{10.3389/fphy.2014.00005}.
\begin{align}
H_{total}(\mathbf{M}) &= A\cdot H_\textrm{constraint}(\mathbf{M}) + B\cdot H_\textrm{cost}(\mathbf{M}) \label{eqn:graph_partitioning} 
\end{align}

Figure \ref{fig:bqnet_graph_partitioning} shows the binary quadratic net for the graph partitioning example, where the subnet composed of the places and the transitions connected by the solid arcs correspond to $H_\textrm{cost}(\mathbf{M})$, and the other subnet with all places and transitions connected by the dotted arcs are added based on $H_\textrm{constraint}(\mathbf{M})$.

One of the feasible solutions is such that $M(\hat{p}_1)=M(\hat{p}_2)=M(\hat{p}_6)=M(\hat{p}_7)=+1$ and $M(\hat{p}_3)=M(\hat{p}_4)=M(\hat{p}_5)=M(\hat{p}_8)=-1$.

\begin{figure}[h]
 \begin{center}
 \includegraphics[scale=0.3]{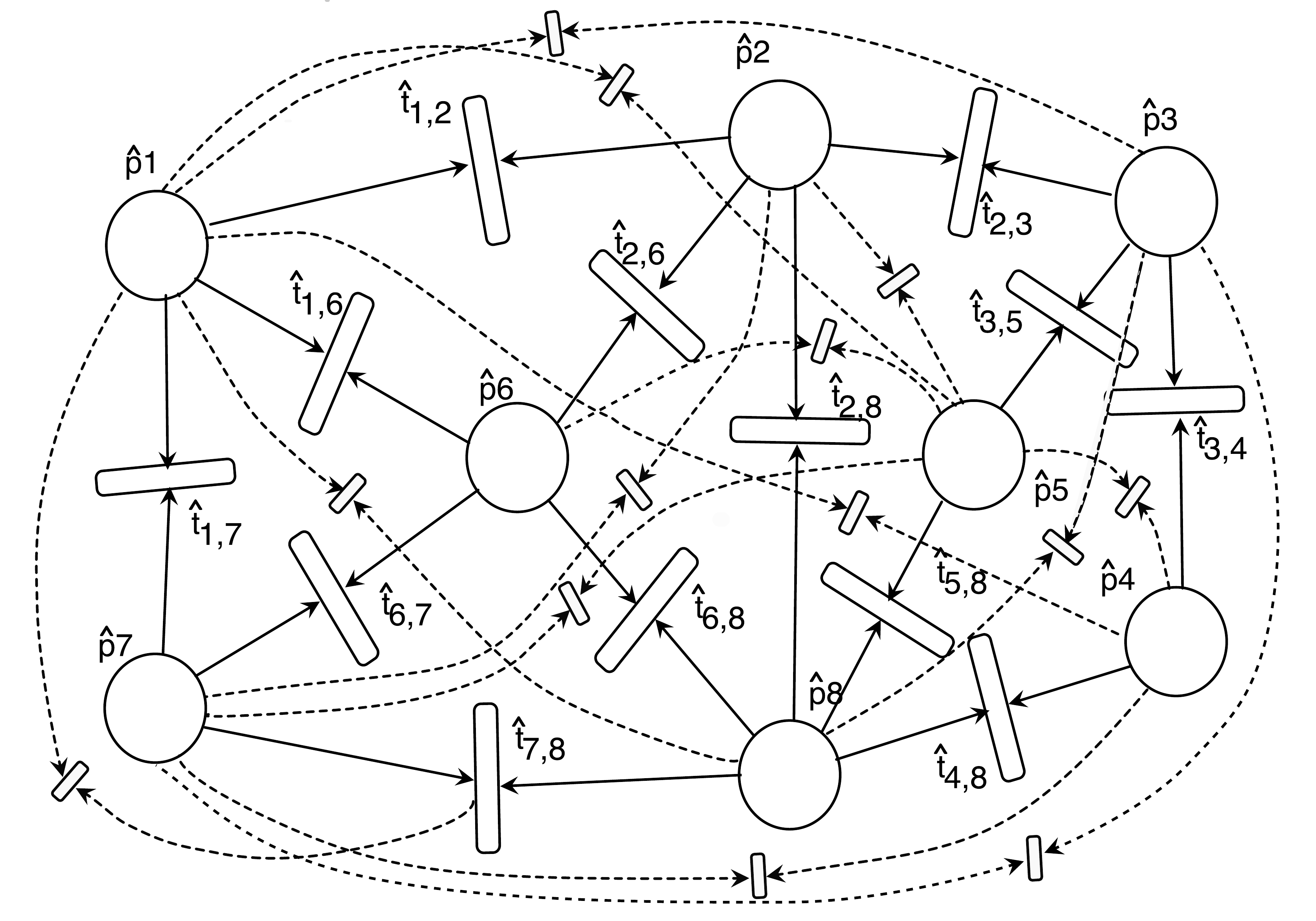}
 \caption{Binary Quadratic Net of Graph Partitioning Problem for Instance shown in Fig.\ref{fig:undirected_graph}}
 \label{fig:bqnet_graph_partitioning}
 \end{center}
\end{figure}
\end{Example}

%%%%%%%%%%%%%%%%%%%%%%%%%%%%%%%%%%%%%%%%%%
\section{Binary Quadratic Net Construction from Problem Domain Petri Nets}

In the previous section, we modeled combinatorial optimization problems directly, a minimum vertex cover problem, and a graph partitioning problem with binary quadratic nets and formulated them as marking problems.
However, they are exceptional cases.
In general, we can meet the conceptual gap between combinatorial optimization problems and target binary quadratic nets.
Note that the problem is more serious in the direct formulation of the Ising or QUBO models.

Our approach attempts to minimize the gap by converting problem-domain Petri net models, represented by timed Petri nets and colored Petri nets, into the corresponding binary quadratic nets.
This section proposes a method for constructing target binary quadratic nets from problem-domain Petri net models.

\subsection{Incremental Construction based on Superposition Principle}

A binary quadratic net can be composed incrementally by combining binary quadratic subnets, each corresponding to a constraint or an objective function in the original optimization problem.
This superposition principle of the net structure and weight values simplifies the binary quadratic net construction. 
The composition is straightforward, where we add the weight values on the places and transitions if the same places and transitions in the subnets are to be combined; otherwise, add new places and transitions with their weight values.
In the previous section, we observed this process in the graph partitioning example (Example 2). Even though we did not consider the weight values, we combined the two subnets.
The following definition formally represents the binary quadratic net composition.
\begin{Definition}
\textbf{(Binary Quadratic Net Composition)} 
For two given binary quadratic nets $\textrm{BQN}_h$ = ($\hat{\Sigma}$, $\hat{P}_h$, $\hat{T}_h$, $\hat{F}_h$, $\hat{C}_h$, $\hat{w}_h)$ and $\textrm{BQN}_k = (\hat{\Sigma}, \hat{P}_k, \hat{T}_k, \hat{F}_k, \hat{C}_k, \hat{w}_k)$ such that the model types of both nets, Ising or QUBO, are the same, the new binary quadratic net $\textrm{BQN} = (\hat{\Sigma}$, $\hat{P}$, $\hat{T}$, $\hat{F}$, $\hat{C}$, $\hat{w})$ is composed based on the following superposition principle:
\begin{align}
 \hat{\Sigma} &\in \{\{-1, +1\}, \{0, 1\}\} \\
\hat{P} & = \hat{P}_h \cup \hat{P}_k \label{eqn:compositionP}\\
\hat{T} &= \hat{T}_h \cup \hat{T}_k \label{eqn:compositionT}\\
\hat{F} &: (\hat{P} \times \hat{T})\cup (\hat{T}\times \hat{P}) \rightarrow 
\begin{cases}
 1 & (\hat{p}_i, \hat{t}_{i,j}) \textrm{ and } (\hat{p}_j, \hat{t}_{i,j}), \hat{t}_{i,j} \in \hat{T}, \\
 0 & \textrm{otherwise}
\end{cases}  \label{eqn:compositionArc}\\
\hat{C}: & \hat{P} \rightarrow 
\begin{cases}
 \{-1, +1\} & \textrm{Ising\ Model}\\
 \{0, 1\} & \textrm{QUBO\ Model}
\end{cases} \\
\hat{w} &: \hat{P}\cup \hat{T} \rightarrow 
\begin{cases}
 w_h(x) + w_k(x) & x \in \hat{P}_h \cup \hat{T}_h \textrm{ and } x \in \hat{P}_k \cup \hat{T}_k, \\
 w_h(x) & x \in \hat{P}_h \cup \hat{T}_h \textrm{ and } x \notin \hat{P}_k \cup \hat{T}_k, \\
 w_k(x) & x \in \hat{P}_k \cup \hat{T}_k \textrm{ and } x \notin \hat{P}_h \cup \hat{T}_h, \\
 0 & \textrm{otherwise}
\end{cases} \label{eqn:compositionW}
\end{align}
\label{def:bqnComposition}
\end{Definition}

The following proposition is a straightforward but essential property for validating our method.
\begin{Proposition}\textbf{(Superposition Property)}
Let $\textrm{BQN}$ be composed from $\textrm{BQN}_h$ and $\textrm{BQN}_k$.
The following properties hold.
\begin{align}
H_{BQN}(\mathbf{M}) &= H_{BQN_h}(\mathbf{M}) +  H_{BQN_k}(\mathbf{M}), & (Additivity), \label{eqn:lineality1}\\
H_{A\cdot BQN}(\mathbf{M}) &= A\cdot H_{BQN}(\mathbf{M}), &(Homogeneity), \label{eqn:lineality2}
\end{align}
where $A$ is a scalar, and $H_{A\cdot BQN}(\mathbf{M})$ is the energy function of $\textrm{BQN}$, such that we replace the weight function $\hat{w}$ with $A\cdot \hat{w}$.
\label{prop:superposition}
\end{Proposition}
\begin{proof}
In Definition \ref{def:bqnComposition}, the energy function is defined by the net structure and weight function $\hat{w}$.
Based on the composition rules in (\ref{eqn:compositionP}), (\ref{eqn:compositionT}), and (\ref{eqn:compositionArc}), all structural properties of $\textrm{BQN}_h$ and $\textrm{BQN}_k$ are transformed into $\textrm{BQN}$.
Based on the rule for the weight function (\ref{eqn:compositionW}), we can confirm that $\hat{w}(\hat{p}_{i})$ and $\hat{w}(\hat{t}_{i, j})$ can be divided into $\hat{w}_h(\hat{p}_{i})$ + $\hat{w}_k(\hat{p}_{i})$ and $\hat{w}_h(\hat{t}_{i, j})$ + $\hat{w}_k(\hat{t}_{i, j})$ for the common place $\hat{p}_{i}$ and transition $\hat{t}_{i, j}$, respectively.
Therefore, the additivity in (\ref{eqn:lineality1}) holds.

The homogeneity property with a degree of 1 is given by Definition \ref{def:enegyFunction} if we replace the weight function $\hat{w}$ with $A\cdot \hat{w}$.
\end{proof}

Owing to the additivity and homogeneity properties in Proposition \ref{prop:superposition}, we have the following corollary:

\begin{Corollary}
Binary quadratic nets can be composed by the incremental application of Definition \ref{def:bqnComposition}, in which we can scale the weight function $\hat{w}$ in subnets with a constant factor.
\end{Corollary}

In our approach, we construct a target binary quadratic net based on the incremental compositions of binary quadratic subnets.
Each subnet is converted from a property of the Petri net model representing the optimization problems. We call the Petri net \textit{problem-domain Petri net}.
The properties of the problem-domain Petri net models are expressed with marking or firing sequences.
To focus on markings or firing sequences, we employ marking-based or firing-based constructions.
In the following subsections, we assume that binary quadratic nets are QUBO models unless otherwise stated, but can be converted into Ising models by using the conversion rule (\ref{eqn:ising_qubo}).

\subsection{Marking-based Construction}

Let us denote a problem-domain Petri net for a target optimization problem by $\mathit{N} = (P, T, F)$ with $P=\{P_1, P_2, ..., P_n\}$ and $T=\{T_1, T_2, ..., T_m\}$.
Marking $M_k(P_j)$ represents a set of tokens in place $P_j$ at time step $k$.
In addition, $\mathbf{M}_k$ is a vector of size $|P|$, showing a token distribution on the Petri net at time step $k$.
The marking trajectory of a Petri net model, $\mathbf{M}_0, \mathbf{M}_1, ..., \mathbf{M}_K$ denotes the status changes triggered by the firing of transitions. 

In the marking-based construction of binary quadratic nets, we represent a marking trajectory of the problem domain Petri net by the place set of the target binary quadratic net.
If each place $P_i$ in the problem-domain Petri net has at most one token at any time and the maximum step is $K$, we initially prepare the following place set of the target binary quadratic net.
\begin{align}
\hat{P} &= \{\hat{p}_i^k | P_i \in P, k=0, 1, 2, ..., K\}
\end{align}
If $M_k(P_i)$ is a natural number, that is, more than one token or a color token with a natural number value is possible in each place in $P$, we need to prepare more places because each place in the binary quadratic nets has one token with color in $\{0, 1\}$ or $\{-1, +1\}$.
\begin{align}
\hat{P} &= \{\hat{p}_{i, n}^k | P_i \in P, n = 0, 1, 2, ..., N, k=0, 1, 2, ..., K\},
\end{align}
where $N$ is the possible maximum value of $M_k(P_{i})$.
At the same time, we require a one-hot constraint because only one place between $\hat{p}_{i, n}^k$ for $n=0, 1, ..., N$ for each $i$ and $k$ should be marked with $1$, and the others should be marked with $0$:
\begin{align}
\sum_{n=0}^{N}M(\hat{p}_{i, n}^k) & = 1, \forall i, k
\end{align}
As the energy function representation, we need the following:
\begin{align}
H_{onehot}(\mathbf{M}) =(\sum_{n=0}^{N}M(\hat{p}_{i, n}^k) & - 1)^2, \forall i, k \label{eqn:onehot}
\end{align}
Note that we can obtain the corresponding binary quadratic net from an energy function.

\subsubsection{Boundedness}
The boundedness is an essential characteristic to ensure avoiding an overflow in the system behavior.
In the Petri net theory, we can express this characteristic as follows.
\begin{equation}
M_k(P_i) \le U_i, \forall i, k \label{eqn:upperbound} 
\end{equation}
where $U_i$ is the upper bound for $P_i$.
We can convert the constraint into a binary quadratic net and represent it as the energy function under the one-hot constraint (\ref{eqn:onehot}):
\begin{align}
H_{boundedness}(\mathbf{M}) &= \sum_{k=0}^{K+1} \sum_{i=1}^{|P|} (\sum_{n=0}^{N}n M(\hat{p}_{i, n}^k) - U_i)^2, \label{eqn:upperbound1}
\end{align}
Function (\ref{eqn:upperbound1}) is sufficient for the equality constraint $M_k(P_i) = U_i$, but not for the upper bound.
Therefore, we improve the function by introducing ancilla places, $\hat{u}_{i, m}, i=1, 2, ..., ,|P|$, $m=0, 1, ..., U_i$. 
This technique is well known in the Ising model formulation  \citep{10.3389/fphy.2014.00005}, \citep{pyqubo}.
\begin{align}
H_{boundedness}(\mathbf{M}) &= \sum_{k=0}^{K+1} \sum_{i=1}^{|P|} (\sum_{n=0}^{N}n M(\hat{p}_{i, n}^k) - \sum_{m=0}^{U_i} m \hat{u}_{i, m})^2 + (\sum_{m=0}^{U_i}{\hat{u}_{i, m}} -1)^2\label{eqn:upperbound2}
\end{align}
The upper bound constraint appears in numerous optimization problems.
The knapsack constraint is a well-known example of this constraint for the knapsack place.
We can also express the boundedness of specific places $P_i$ by removing the other places from the function.

\subsubsection{Invariant}

The boundedness shown in (\ref{eqn:upperbound}) denotes inequality constraints.
The invariance leads to equality constraints based on markings.
Note that the invariance is different from the structural invariance of the net theory.
The behavioral invariance requires that the total weighted sum of the tokens be equal among the firing sequences.
The sum may be the cost required for the resource to operate the system.
The following constraint shows that the weighted sum becomes $W$ for all $k$.
\begin{equation}
\sum_{i=1}^{|P|} h_i M_k(P_i) = W, k=0, 1, ..., K. \label{eqn:invariant} 
\end{equation}

We can convert the constraint into a binary quadratic net and represent it as the energy function under the one-hot constraint (\ref{eqn:onehot}):
\begin{align}
H_{invariant}(\mathbf{M}) &= \sum_{k=0}^{K+1} (\sum_{i=1}^{|P|} \sum_{n=0}^{N}h_i n M(\hat{p}_{i, n}^k) - W)^2) \label{eqn:invariant1}
\end{align}

\subsection{Firing-based Construction}

A firing count vector $\mathbf{X}_k$ represents the firing counts of each transition in a Petri net model at step $k$, where $X_k(T_j)$ denotes the firing counts of transition $T_j$ at time step $k$.
Elements of $\mathbf{X}_k$ are usually one of $\{0, 1\}$; that is, transitions can fire only once at each time step.
We call this single firing restriction \textit{single-server semantics} in Petri net theory. However, any natural number values (allowing more than one) are also possible with \textit{infinite server semantics}. 
For the single-server semantics, we generate a place set $\hat{P}$ in the target binary quadratic net, where $M(\hat{p}_i^k)$ corresponds to $X_k(T_i)$:
\begin{align}
\hat{P} &= \{\hat{p}_i^k | T_i \in T, k=0, 1, 2, ..., K\}
\end{align}
Similarly, we can extend the single-firing semantics to the $N$ times firing semantics.
Note that an infinite number of firing times is possible mathematically but impossible practically. 
Thus, we restrict the maximum number reasonably to $N$.
\begin{align}
\hat{P} &= \{\hat{p}_{i, n}^k | T_i \in P, n = 0, 1, 2, ..., N, k=0, 1, 2, ..., K\},
\end{align}
As the energy function representation, we need the same one-hot constraint (\ref{eqn:onehot}):

\subsubsection{Resource Conflict}

If transitions $T_i$ and $T_j$ have a common input place and conflict with a single token, $X_k(T_i) = X_k(T_j) =1$ cannot be allowed.
Therefore, we need the following energy function:
 \begin{align}
H_{conflict}(\mathbf{M}) &= \sum_{k=0}^{K}\sum_{(T_i, T_j) \in C} M(\hat{p}_{i}^k)M(\hat{p}_{j}^k), \label{eqn:conflict}
\end{align}
where $C$ is a set of conflict transitions. In addition, $C$ can be extracted from the problem domain Petri net. 

In timed Petri net models, we need to consider the firing duration.
Let $\mathit{N} = (P, T, F, TS, FD)$ with $P=\{P_1, P_2, ..., P_n\}$ and $T=\{T_1, T_2, ..., T_m\}$ be a timed Petri net, where $FD:T\rightarrow \mathbb{N}$ is a function that returns the firing duration.
We then extend the resource conflict in the stepwise firing into the timed firing. 
 \begin{align}
H_{conflict}(\mathbf{M}) &= \sum_{(T_i, T_j, k, h) \in C^{timed}} M(\hat{p}_{i}^k)M(\hat{p}_{j}^h),\label{eqn:conflict2} \\
C^{timed} &= \{(T_i, T_j, k, h)|\forall (T_i, T_j) \in C, h \le k+FD(T_i)\ \mathrm{or}\ k\le h+ FD(T_j)\}, \label{eqn:conflict3} 
\end{align}
where $C^{timed}$ is the timed conflict set such that conflict transitions cannot fire until the firing duration is complete.
We can obtain $C^{timed}$ from the given timed Petri net.

\subsubsection{Firing Count}
In some applications, we must specify the number of firing occurrences for each transition.

\begin{equation}
\sum_{k=0}^{K}X_k(T_i) = FC_i, \forall i \label{eqn:firingcount} 
\end{equation}
where $FC_i$ is the specified number of firings of $T_i$.

We can convert the constraint into a binary quadratic net and represent it as the energy function under the one-hot constraint (\ref{eqn:onehot}):
\begin{align}
H_{firings}(\mathbf{M}) &= \sum_{i=1}^{|T|} ( \sum_{k=0}^{K}\sum_{n=0}^{N}n M(\hat{p}_{i, n}^k) - FC_i)^2) \label{eqn:finringcount2}
\end{align}
Note that $M(\hat{p}_{i, n}^k)$ in the binary quadratic net corresponds to $X_k(T_i)$ in the problem domain Petri net.

Assuming that each transition $T_i$ should fire exactly once during $\mathbf{X}_0$, $\mathbf{X}_1$, ..., $\mathbf{X}_{K}$, the function (\ref{eqn:firingcount}) can be represented as follows:
\begin{align}
H_{firings}(\mathbf{M}) &= \sum_{i=1}^{|T|} (\sum_{k=0}^{K}M(\hat{p}_{i}^k) - 1)^2 \label{eqn:finringcount3}
\end{align}
In practical cases, this constraint is commonly used.

\subsubsection{Precedence Relation}
Let us assume again that each transition $T_i$ should fire exactly once during $\mathbf{X}_0$, $\mathbf{X}_1$, ..., $\mathbf{X}_{K}$.
We consider the precedence relation between the firing of transitions.
If $T_i$ precedes structurally to $T_j$, that is, $T_i^\bullet \subseteq\ ^\bullet T_j$ and $^\bullet T_i \not\subseteq T_j^\bullet$, the following penalty function is required.
 \begin{align}
H_{precedence}(\mathbf{M}) &= \sum_{k=0}^{K}\sum_{(T_i, T_j) \in Prec, h\le k} M(\hat{p}_{i}^k)M(\hat{p}_{j}^h), \label{eqn:precedence}
\end{align}
where $Prec$ is the set of precedence relations. Here, $Prec$ can be extracted from the problem domain Petri net. 
Similar to the firing conflict, we consider timed Petri nets by introducing the firing duration.

 \begin{align}
 H_{precedence}(\mathbf{M}) &= \sum_{k=0}^{K}(\sum_{(T_i, T_j) \in Prec, h\le k+FD(T_i)} M(\hat{p}_{i}^k)M(\hat{p}_{j}^{h})) \label{eqn:precedence2}
 \end{align}

\subsection{Application Example 1 (Marking-based Construction): Traveling Salesman Problems}
We consider the traveling salesman problem as an example of a marking-based construction.
Let $\mathit{N} = (P, T, F, TS, FD)$ with $P=\{P_1, P_2, ..., P_n\}$ and $T=\{T_{i, j}| \forall (P_i, P_j) \in P\times P \}$ be a timed Petri net, where $FD:T\rightarrow \mathbb{N}$ is a function that returns the firing duration.
Each place $P_i$ corresponds to a city, and the transition $T_{i. j}$ denotes the path from $P_i$ to $P_j$. 
The firing of $T_{i, j}$ indicates the movement from $P_i$ to $P_j$. 
The initial marking $\mathbf{M}_0$ includes only one token at $P_1$.

Because the salesman visits each city only once from the problem definition,
\begin{align}
\sum_{k=0}^{|P|-1} M_k(P_i) &= 1, \forall P_i \in P.
\end{align}
In addition, because the salesman should be at one place for each step,
\begin{align}
\sum_{i=1}^{|P|} M_k(P_i) &= 1, k=0, 1, ..., |P|-1.
\end{align}
The objective function is to minimize the total time to visit all cities and return to the starting place.
\begin{align}
distance & = \sum_{i=1}^{|P|} \sum_{j=1}^{|P|} (FD(T_{i,j})\cdot \sum_{k=0}^{|P|-1}M_k(P_i)M_{k+1}(P_j))
\end{align}
By converting the constraints and objective function in the problem-domain Petri net into the binary quadratic nets, we have
\begin{align}
H_{visitingOnce}(\mathbf{M}) &= \sum_{i=1}^{|P|} (\sum_{k=0}^{|P|-1} M(\hat{p}_i^k)-1)^2 \\
H_{singleness}(\mathbf{M}) &= \sum_{k=0}^{|P|-1} (\sum_{i=1}^{|P|} M(\hat{p}_i^k) -1)^2\\
H_{distance}(\mathbf{M}) & = \sum_{i=1}^{|P|} \sum_{j=1}^{|P|} (FD(T_{i,j})\cdot \sum_{k=0}^{|P|-1}M(\hat{p}_i^k)M(\hat{p}_j^{k+1})).
\end{align}
The total binary quadratic net is as follows:
\begin{align}
H_{total}(\mathbf{M}) &= A\cdot H_{visitingOnce}(\mathbf{M}) + B\cdot H_{singleness}(\mathbf{M}) + C \cdot H_{distance}(\mathbf{M}),
\end{align}
where $A$, $B$, and $C$ denote the scale factors of the corresponding subnets.
The formulation is equivalent to the QUBO model in \citep{10.3389/fphy.2014.00005}.

\subsection{Application Example 2 (Firing-based Construction): Job-shop Scheduling Problems}
We implemented our method by incorporating CPNTools \citep{cpn}, SNAKES\citep{10.1007/978-3-319-19488-2_13}, and PyQUBO \citep{pyqubo}.
As an example, we modeled a simple job-shop scheduling problem with three jobs, four tasks per job, and three shared resources.
Figure \ref{fig:jobshop} shows the model drawn by the GUI software of CPNTools.
Three sequential systems, 
\begin{align}
&p0, t0, p1, t1, p2, t2, p3, t3, p4, \nonumber \\
&p5, t4, p6, t5, p7, t6, p8, t7, p9, \nonumber \\
&p10, t8, p11, t9, p12, t10, p13, t11, p14,\nonumber 
\end{align} represent the jobs, and places $m0, m1, m2$ represent the resources.
Our software reads the model, represents the problem-domain Petri net model with SNAKES, and then extracts the necessary information to convert it into the target binary quadratic net.

The following are the energy functions of the binary quadratic subnets corresponding to the precedence relation constraint between tasks in each job, the resource conflict constraint for each shared resource, and the firing count constraint for each task.
 
 \begin{align}
 H_{precedence}(\mathbf{M}) &= \sum_{k=0}^{MaxTime}(\sum_{(T_i, T_j) \in Prec, h\le k+FD(T_i)} M(\hat{p}_{i}^k)M(\hat{p}_{j}^{h})) \label{eqn:jobshopprecedence}\\
H_{conflict}(\mathbf{M}) &= \sum_{(T_i, T_j, k, h) \in C^{timed}} M(\hat{p}_{i}^k)M(\hat{p}_{j}^h), \label{eqn:jobshopconflict}\\
H_{firings}(\mathbf{M}) &= \sum_{i=1}^{|T|} (\sum_{k=0}^{MaxTime}M(\hat{p}_{i}^k) - 1)^2. \label{eqn:jobshopfinringcount}
 \end{align}

The precedence set $Prec$ in (\ref{eqn:jobshopprecedence}) is easily constructed by extracting the precedence relation in each sequential system (job).
We can extract the timed conflict set $C^{timed}$ in (\ref{eqn:jobshopconflict}) from the connection between the transitions and resource places and the firing duration $FD$.
\textit{MaxTime} is the only parameter we need to set before the optimization process. 
This value indicates the delivery time deadline.
The total binary quadratic net is as follows:
\begin{align}
H_{total}(\mathbf{M}) &= A\cdot H_{precedence}(\mathbf{M}) + B\cdot H_{conflict}(\mathbf{M}) + C \cdot H_{firings}(\mathbf{M})
\end{align}
Note that this formulation is used to obtain a feasible solution; however, we can extend this to the optimization by incorporating it with a binary search to find the minimum \textit{MaxTime}.
This formulation is equivalent to the QUBO model in \citep{venturelli2016quantum}.
Figure \ref{fig:jobshopResult} shows the schedule obtained through the annealing process using a Fujitsu Digital Annealer.

\begin{figure}[tbhp]
 \begin{center}
 \includegraphics[scale=0.4]{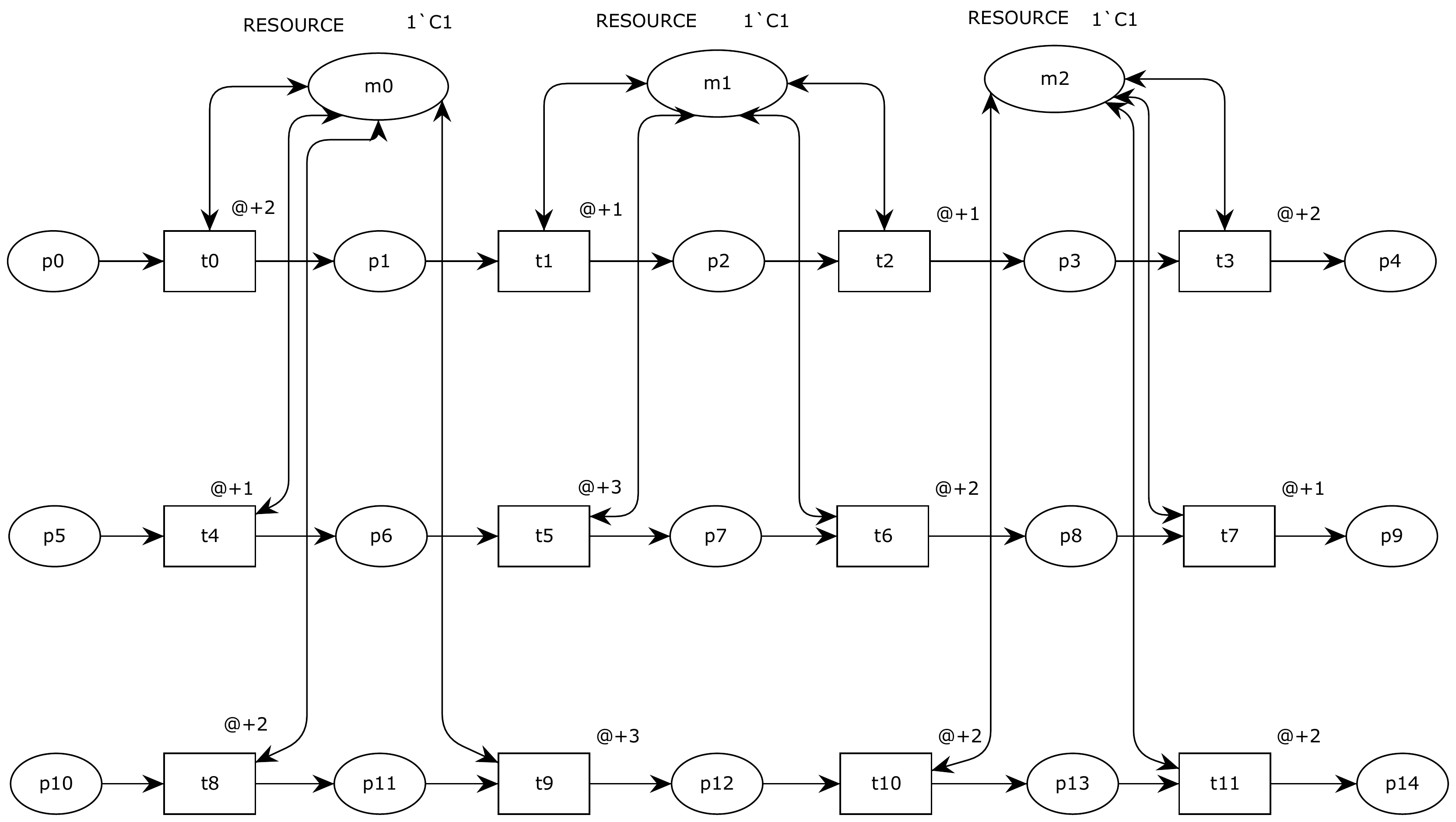}
 \caption{Colored Timed Petri Net Model for Job Shop Scheduling drawn by CPNTools \citep{cpn}}
 \label{fig:jobshop}
 \end{center}
\end{figure}

\begin{figure}[tbhp]
 \begin{center}
 \includegraphics[scale=0.7]{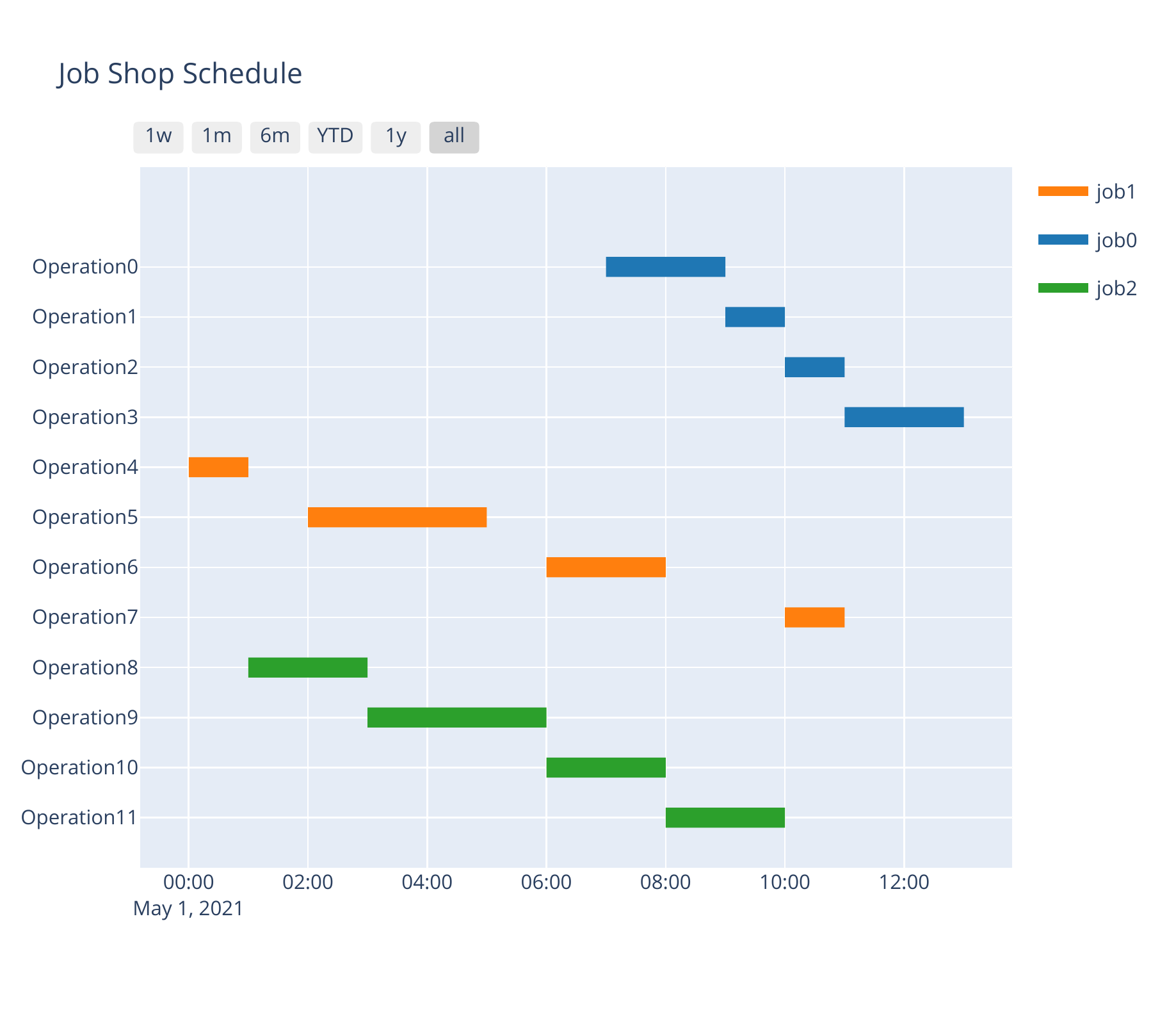}
 \caption{Job Shop Schedule obtained with Fujitsu Digital Annealer}
 \label{fig:jobshopResult}
 \end{center}
\end{figure}

%%%%%%%%%%%%%%%%%%%%%%%%%%%%%%%%%%%%%%%%%%
\section{Conclusions}

This paper proposes a Petri net modeling approach to the Ising model formulation for quantum annealing.
Although our method requires users to model their optimization problems with Petri nets, this process can be carried out in a relatively straightforward manner if we know the target problem and the simple Petri net modeling rules.
Therefore, we can drastically relax the difficulty of the Ising model formulation.
We implemented our method with Python incorporated using well-known Petri net tools, CPNTools and SNAKES.
We can automatically generate the Ising models for optimization problems, such as scheduling problems, vehicle routing problems, portfolio optimization problems, and others once we model the target optimization problems with Petri nets.

We defined binary quadratic nets to represent the Ising model formulation.
However, the binary quadratic net can also be used to analyze the quantum annealing process by attaching an additional subnet that simulates the annealing.
This tool may contribute to parameter tuning for annealing, which is another task used to expand the quantum annealing technology mentioned in the Introduction.

\appendix
\renewcommand{\thetable}{\Alph{section}.\arabic{table}}
\renewcommand{\thefigure}{\Alph{section}.\arabic{figure}}

\section{Apendix}
%\subsection{Interaction Primitives}
Table \ref{atbl:quboip} shows the interaction primitives $I^{qubo}_i, i=0, 1,..., 15$ for two binary variables in $\{0, 1\}$, 
where $(0, 0), (0, 1), (1, 0)$, and $(1, 1)$ express all possible combinations of the two binary variables.
In addition, $I^{qubo}_1, I^{qubo}_7,  I^{qubo}_8,  I^{qubo}_9$ correspond to AND, XOR, OR, and NOR, respectively.
Moreover, $I^{qubo}_{0}, I^{qubo}_{15}$ are the inconsistency and tautology, respectively.
The others also show important logical functions.

Table \ref{atbl:isingip} shows the interaction primitives for the Ising model converted from Table \ref{atbl:quboip} by using the following relation.
\begin{align}
M(\hat{p}_i)^{ising} &= 2 M(\hat{p}_i)^{qubo} -1  \label{eqn:ising_qubo} \\ 
M(\hat{p}_i)^{qubo} &= (M(\hat{p}_i)^{ising} + 1)/2 \label{eqn:qubo_ising} 
\end{align}

%\begin{specialtable}[h]
\begin{table}[h]
\begin{centering}
\caption{Interaction Primitives in QUBO Models}
\label{atbl:quboip}
\begin{tabular}{c||c|c|c|c||l}
%\toprule

\hline\hline
  & (0, 0) & (0, 1) & (1, 0) & (1, 1) & \multicolumn{1}{c}{Energy function }            \\ \hline\hline
$I^{qubo}_0$                & 0      & 0      & 0      & 0      & $0$                       \\ \hline
$I^{qubo}_1$                & 0      & 0      & 0      & 1      & $M(\hat{p}_i)M(\hat{p}_j)$                       \\ \hline
$I^{qubo}_2$                   & 0      & 0      & 1      & 0      & $M(\hat{p}_i)(1-M(\hat{p}_j))$                       \\ \hline
$I^{qubo}_3$                   & 0      & 0      & 1      & 1      & $M(\hat{p}_i)$                       \\ \hline
$I^{qubo}_4$                   & 0      & 1      & 0      & 0      & $(1-M(\hat{p}_i))M(\hat{p}_j)$                       \\ \hline
$I^{qubo}_5$                   & 0      & 1      & 0      & 1      & $M(\hat{p}_j)$                       \\ \hline
$I^{qubo}_6$                & 0      & 1      & 1      & 0      & $M(\hat{p}_i) + M(\hat{p}_j) - 2M(\hat{p}_i)M(\hat{p}_j)$ \\ \hline
$I^{qubo}_7$                 & 0      & 1      & 1      & 1      & $M(\hat{p}_i) + M(\hat{p}_j) - M(\hat{p}_i)M(\hat{p}_j)$                 \\ \hline
$I^{qubo}_8$                & 1      & 0      & 0      & 0      & $1-M(\hat{p}_i) - M(\hat{p}_j) + M(\hat{p}_i)M(\hat{p}_j)$        \\ \hline
$I^{qubo}_9$               & 1      & 0      & 0      & 1      & $1-M(\hat{p}_i)-M(\hat{p}_j) + 2M(\hat{p}_i)M(\hat{p}_j)$ \\ \hline
$I^{qubo}_{10}$               & 1      & 0      & 1      & 0      & $1-M(\hat{p}_j)$ \\ \hline
$I^{qubo}_{11}$               & 1      & 0      & 1      & 1      & $1 - M(\hat{p}_j) + M(\hat{p}_i)M(\hat{p}_j)$ \\ \hline
$I^{qubo}_{12}$              & 1      & 1      & 0      & 0      & $1 - M(\hat{p}_i)$ \\ \hline
$I^{qubo}_{13}$               & 1      & 1      & 0      & 1      & $1 - M(\hat{p}_i) + M(\hat{p}_i)M(\hat{p}_j)$ \\ \hline
$I^{qubo}_{14}$               & 1      & 1      & 1      & 0      & $1 - M(\hat{p}_i)M(\hat{p}_j)$                         \\ \hline
$I^{qubo}_{15}$               & 1      & 1      & 1      & 1      & $1$                         \\ 
%\bottomrule
\end{tabular}
\\ $x$ and $y$ in $(x, y), x, y \in \{0, 1\}$, show $M(\hat{p}_i)$ and $M(\hat{p}_j)$, respectively.\\
A value of 1 and 0 in each cell represents a preferable and an un-preferable interaction, respectively.
%\end{specialtable}
\end{centering}
\end{table}

%\begin{specialtable}[h]
\begin{table}[h]
\begin{centering}
\caption{Interaction Primitives in Ising Models}
\label{atbl:isingip}
{\small
\begin{tabular}{c||c|c|c|c||l}
%\toprule
%\hline
  & (-1, -1) & (-1, +1) & (+1, -1) & (+1, +1) & \multicolumn{1}{c}{Energy function }            \\ \hline\hline
$I^{Ising}_0$                & 0      & 0      & 0      & 0      & $0$                       \\ \hline
$I^{Ising}_1$                & 0      & 0      & 0      & 1      & $\frac{1}{4}(M(\hat{p}_i)+1)(M(\hat{p}_j)+1)$                       \\ \hline
$I^{Ising}_2$                   & 0      & 0      & 1      & 0      & $\frac{1}{4}(M(\hat{p}_i)+1)(-M(\hat{p}_j)+1)$                       \\ \hline
$I^{Ising}_3$                   & 0      & 0      & 1      & 1      & $\frac{1}{2}(M(\hat{p}_i)+1)$                       \\ \hline
$I^{Ising}_4$                   & 0      & 1      & 0      & 0      & $\frac{1}{4}(-M(\hat{p}_i)+1)(M(\hat{p}_j)+1)$                       \\ \hline
$I^{Ising}_5$                  & 0      & 1      & 0      & 1      & $\frac{1}{2}(M(\hat{p}_j)+1)$                       \\ \hline
$I^{Ising}_6$                & 0      & 1      & 1      & 0      & $\frac{1}{2}(1- M(\hat{p}_i)M(\hat{p}_j))$ \\ \hline
$I^{Ising}_7$                 & 0      & 1      & 1      & 1      &  $\frac{1}{4}(M(\hat{p}_i)+M(\hat{p}_j)-M(\hat{p}_i)M(\hat{p}_j)+3)$                \\ \hline
$I^{Ising}_8$                & 1      & 0      & 0      & 0      & $\frac{1}{4}(-M(\hat{p}_i)-M(\hat{p}_j)+M(\hat{p}_i)M(\hat{p}_j)+1)$        \\ \hline
$I^{Ising}_9$               & 1      & 0      & 0      & 1      & $\frac{1}{2}(M(\hat{p}_i)M(\hat{p}_j)+1)$  \\ \hline
$I^{Ising}_{10}$               & 1      & 0      & 1      & 0      & $\frac{1}{2}(1-M(\hat{p}_j))$ \\ \hline
$I^{Ising}_{11}$               & 1      & 0      & 1      & 1      & $\frac{1}{4}(M(\hat{p}_i)-M(\hat{p}_j)+M(\hat{p}_i)M(\hat{p}_j)+3)$ \\ \hline
$I^{Ising}_{12}$               & 1      & 1      & 0      & 0      & $\frac{1}{2}(1-M(\hat{p}_i))$ \\ \hline
$I^{Ising}_{13}$               & 1      & 1      & 0      & 1      & $\frac{1}{4}(-M(\hat{p}_i)+M(\hat{p}_j)+M(\hat{p}_i)M(\hat{p}_j)+3)$ \\ \hline
$I^{Ising}_{14}$               & 1      & 1      & 1      & 0      & $\frac{1}{4}(-M(\hat{p}_i)-M(\hat{p}_j)-M(\hat{p}_i)M(\hat{p}_j)+3)$                         \\ \hline
$I^{Ising}_{15}$               & 1      & 1      & 1      & 1      & $1$    \\                    
%\bottomrule
\end{tabular}
}
\\ $x$ and $y$ in $(x, y), x, y \in \{-1, +1\}$, show $M(\hat{p}_i)$ and $M(\hat{p}_j)$, respectively.\\
A value of 1 and 0 in each cell represents a preferable and an un-preferable interaction, respectively.
%\end{specialtable}
\end{centering}
\end{table}

\bibliographystyle{unsrtnat}
\bibliography{pn2qubo_bib}

%=====================================
% References, variant B: internal bibliography
%=====================================

% If authors have biography, please use the format below
%\section*{Short Biography of Authors}
%\bio
%{\raisebox{-0.35cm}{\includegraphics[width=3.5cm,height=5.3cm,clip,keepaspectratio]{Definitions/author1.pdf}}}
%{\textbf{Firstname Lastname} Biography of first author}
%
%\bio
%{\raisebox{-0.35cm}{\includegraphics[width=3.5cm,height=5.3cm,clip,keepaspectratio]{Definitions/author2.jpg}}}
%{\textbf{Firstname Lastname} Biography of second author}

% The following MDPI journals use author-date citation: Arts, Econometrics, Economies, Genealogy, Humanities, IJFS, JRFM, Laws, Religions, Risks, Social Sciences. For those journals, please follow the formatting guidelines on http://www.mdpi.com/authors/references
% To cite two works by the same author: \citepauthor{ref-journal-1a} (\citepyear{ref-journal-1a}, \citepyear{ref-journal-1b}). This produces: Whittaker (1967, 1975)
% To cite two works by the same author with specific pages: \citepauthor{ref-journal-3a} (\citepyear{ref-journal-3a}, p. 328; \citepyear{ref-journal-3b}, p.475). This produces: Wong (1999, p. 328; 2000, p. 475)

%%%%%%%%%%%%%%%%%%%%%%%%%%%%%%%%%%%%%%%%%%
%% for journal Sci
%\reviewreports{\\
%Reviewer 1 comments and authors’ response\\
%Reviewer 2 comments and authors’ response\\
%Reviewer 3 comments and authors’ response
%}
%%%%%%%%%%%%%%%%%%%%%%%%%%%%%%%%%%%%%%%%%%
\end{document}